\documentclass{llncs}

\usepackage[utf8]{inputenc}
\usepackage{amsmath}
\usepackage{amsfonts}
\usepackage{amssymb}
\usepackage{vaucanson-g}
\usepackage{graphicx}
\usepackage{graphics}
\usepackage{multirow}
\usepackage{pgfpages}
\usepackage{enumerate}
\usepackage{listings}
\usepackage{array}
\usepackage{tikz}

\title{Incomplete Transition Complexity of Basic Operations on Finite
  Languages \thanks{This work was partially funded by the European
    Regional Development Fund through the programme COMPETE and by the
    Portuguese Government through the FCT under projects
    PEst-C/MAT/UI0144/2011 and CANTE-PTDC/EIA-CCO/101904/2008.}}
\author{Eva Maia\thanks{Eva Maia is funded by FCT grant
    SFRH/BD/78392/2011.}, Nelma Moreira, Rog\'{e}rio Reis}

\institute{CMUP \& DCC, Faculdade de Ciências da Universidade do Porto \\
  Rua do Campo Alegre, 4169-007 Porto, Portugal\\
e-mail:{\tt\{emaia,nam,rvr\}@dcc.fc.up.pt}}

\begin{document}
\maketitle

\begin{abstract}
  The state complexity of basic operations on finite languages
  (considering complete DFAs) has been in studied the literature.
   In this paper we study the incomplete
  (deterministic) state and transition complexity on finite languages
  of boolean operations, concatenation, star, and reversal. For all
  operations we give tight upper bounds for both descriptional
  measures. We correct the published state complexity of concatenation for
  complete DFAs and
  provide a tight upper bound for the case when the \emph{right}
  automaton is larger than the \emph{left} one. For all binary
  operations the tightness is proved using family languages with
  a variable alphabet size. In general the operational complexities
  depend not only on the complexities of the operands but also on
  other refined measures.
\end{abstract}

\section{Introduction}
\label{sec:introduction}
Descriptional complexity studies the measures of complexity of
languages and operations. These studies are motivated by the need to
have good estimates of the amount of resources required to manipulate
the smallest representation for a given language. In general, having
succinct objects will improve our control on software, which may
become smaller and more efficient.  Finite languages are an important
subset of regular languages with many applications in compilers,
computational linguistics, control and verification,
etc.~\cite{owens09:_regul_expres_deriv_re_examin,beesley03:_finit_state_morph,maurel05:_autom_and_diction,cassandras06:_introd_to_discr_event_system}. In
those areas it is also usual to consider deterministic finite automata
(DFA) with partial transition functions. This motivates the study of
the transition complexity of DFAs (not necessarily complete), besides
the usual state complexity. The operational transition complexity of
basic operations on regular languages was studied by Gao \emph{et
  al.}~\cite{gao11:_trans_compl_of_incom_dfas} and Maia \emph{et
  al.}~\cite{maia13:_incom_trans_compl_of_some_basic_operat}.  In this
paper we continue that line of research by considering the class of finite
languages. For finite languages, Salomaa and
Yu~\cite{salomaa97:_nfa_to_dfa_trans_for} showed that the state
complexity of the determinization of a nondeterministic automaton
(NFA) with $m$ states and $k$ symbols is $\Theta(k^{\frac{m}{1+\log
    k}})$ (lower than $2^m$ as it is the case for general regular languages). C\^ampeanu \emph{et
  al.}~\cite{campeanu99:_state_compl_of_basic_operat_finit_languag}
studied the operational state complexity of concatenation, Kleene
star, and reversal. Finally, Han and
Salomaa~\cite{han08:_state_compl_of_union_and} gave tight upper bounds
for the state complexity of union and intersection on finite
languages. In this paper we give tight upper bounds for the transition
complexity of all the above operations. 
 We
correct the upper bound for the state complexity of concatenation
\cite{campeanu99:_state_compl_of_basic_operat_finit_languag},
and show that if the \emph{right} automaton is larger than the
\emph{left} one, the upper bound is only reached using an alphabet of
variable size.  Note
that, the difference between the state complexity for non necessarily
complete DFAs and for complete DFAs is at most one. 
Table~\ref{table:wtc} presents a comparison of the transition
complexity on regular and finite languages, 
where the new results are highlighted. 
All the proofs not
presented in this paper can be found in an extended version of this work\footnote{\url{http://www.dcc.fc.up.pt/Pubs/TReports/TR13/dcc-2013-02.pdf}}.

\begin{table}
\centering
\begin{tabular}{|c||c|c||c|c|}
    \hline
    Operation & Regular&$|\Sigma|$  & Finite & $|\Sigma|$      \\ \hline\hline
    $L_1\cup L_2$ & $2n(m+1)$ & $2$ & $\textbf{3(mn-n-m) +2}$ & $f_1(m,n)$ \\ \hline
    $L_1\cap L_2$ &  $nm$ &$1$  &\multicolumn{1}{p{5.5cm}|}{
      $\mathbf{(m-2)(n-2)(2+\sum_{i=1}^{\min(m,n)-3}(m-}$ $\mathbf{2-i)(n-2-i))+2}$ } 
& $f_2(m,n)$ \\ \hline
    $L^C$ & $m+2$ & $1$ &$\mathbf{m+1}$ & $1$  \\ \hline
    \multirow{2}{*}{$L_1L_2$} & 
$2^{n-1}(6m+3)-5$,& \multirow{2}{*}{$3$} 
&$\mathbf{6.2^{n-1}-8}$, if $m+1\geq n$& $2$\\ \cline{4-5}
&if $m,n \geq 2$ && See Proposition~\ref{prop:cont1} \textbf{(\ref{eq:cont1})}& $n-1$ \\\hline    
\multirow{2}{*}{$L^\star$} & \multirow{2}{*}{$3.2^{m-1}-2$, if $m \geq 2$} 
& \multirow{2}{*}{$2$} 
&\multicolumn{1}{p{5.2cm}|}{$\mathbf{9\cdot2^{m-3}-2^{m/2}-2}$, if $m$ is odd }& \multirow{2}{*}{$3$}\\ \cline{4-4}
&&& $\mathbf{9\cdot 2^{m-3}-2^{(m-2)/2}-2}$, if $m$ is even& \\\hline
    \multirow{2}{*}{$L^R$} & \multirow{2}{*}{$2(2^m-1)$} & \multirow{2}{*}{$2$} & 
   $\mathbf{2^{p+2}-7}$, if $m=2p$ &\multirow{2}{*}{$2$} \\ \cline{4-4}
&&&$\mathbf{3\cdot 2^{p}-8}$, if  $m=2p-1$&  \\\hline
\multicolumn{3}{c}{}
    \end{tabular}
    \caption{Incomplete transition complexity for regular and finite
      languages, where $m$ and $n$ are the (incomplete) state complexities
      of the operands, $f_1(m,n)=(m-1)(n-1)+1$ and $f_2(m,n)= (m-2)(n-2)+1$.}
\label{table:wtc}
\end{table}

\section{Preliminaries}
\label{sec:preliminaries}

We recall some basic notions about finite automata and regular
languages. For more details, we refer the reader to the standard
literature~\cite{hopcroft79:_introd_to_autom_theor_languag_and_comput,yu97:_handb_formal_languag,shallit08:_secon_cours_in_formal_languag}.

Given two integers $m,n\in \mathbb{N}$ let $[m,n]=\{i\in \mathbb{N}\mid m\leq
i \leq n\}$.
A \emph{deterministic finite automaton} (DFA) is a five-tuple
$A=(Q,\Sigma,\delta,q_0,F)$ where $Q$ is a finite set of states,
$\Sigma$ is a finite input alphabet, $q_0\in Q$ is the initial
state, $F \subseteq Q$ is the set of final states, and $\delta$ is the
transition function $Q \times \Sigma \to Q$. 
Let $|\Sigma| = k$, $|Q|=n$, and without lost of generality, we consider $Q=[0,n-1]$ with
$q_0=0$. The transition function 
can be naturally extended to sets in $2^Q$ and to words $w\in \Sigma^\star$. A DFA is
\emph{complete} if the transition function  is total. In
this paper we consider DFAs 
to be not necessarily complete, \emph{i.e.} with partial transition functions.
The \textit{language} accepted by $A$ is $\mathcal{L}(A)$ $=\{w\in
\Sigma^\star\ |\ \delta(0,w) \in F\}$. Two DFAs are
\textit{equivalent} if they accept the same language. For each regular
language, considering or not a total transition function,
there exists a unique minimal complete DFA with a least number of
states. The \textit{left-quotient} of $L\subseteq \Sigma^\star$ by $x
\in \Sigma^\star$ is $D_xL=\{z\ |\ xz \in L \}$. The equivalence
relation $\equiv_L \subseteq \Sigma^\star \times \Sigma^\star$ is
defined by $ x \equiv_L y$ if and only if $D_xL=D_yL$. The
\textit{Myhill-Nerode Theorem} states that a language $L$ is regular
if and only if $\equiv_{L}$ has a finite number of equivalence
classes, \textit{i.e.}, $L$ has a finite number of left quotients.
This number is equal to the number of states of the minimal complete
DFA. The \emph{state complexity}, $sc(L)$, of a regular language $L$
is the number of states of the minimal complete DFA of $L$. If the
minimal DFA is not complete its number of states is the number of left
quotients minus one (the \emph{dead state}, that we denote by
$\Omega$, is removed). The \emph{incomplete state complexity} of a
regular language $L$ ($isc(L)$) is the number of states of the minimal
DFA, not necessarily complete, that accepts $L$. Note that $isc(L)$ is
either equal to $sc(L)-1$ or to $sc(L)$. The \textit{incomplete
  transition complexity}, $itc(L)$, of a regular language $L$ is the
minimal number of transitions over all DFAs that accepts $L$.  We omit
the term \emph{incomplete} whenever the model is explicitly given. 
A $\tau$-\emph{transition} is a transition labeled by $\tau\in\Sigma$.
The $\tau$-\textit{transition complexity} of $L$, $itc_\tau(L)$ is the
minimal number of $\tau$-transitions of any DFA recognizing $L$.  It
is known that $itc(L) = \sum_{\substack{\tau \in \Sigma}}
itc_{\tau}(L)$~\cite{gao11:_trans_compl_of_incom_dfas,maia13:_incom_trans_compl_of_some_basic_operat}.

The \emph{complexity of an operation} on regular languages is
the (worst-case) complexity of a language resulting from the
operation, considered as a function of the complexities of the
operands. Usually an \emph{upper bound} is obtained by
providing an algorithm, which given representations of the operands
(\emph{e.g.} DFAs), constructs a model (\emph{e.g.} DFA)
that accepts the language resulting from the referred operation. 
 To prove that an upper bound is \emph{tight}, for each
operand we can give a family of languages (parametrized by the
complexity measures and called \emph{witnesses}), such that the resulting language achieves that
upper bound.

For determining the transition complexity of an operation, we also
consider the following measures and refined numbers of transitions.
Let $A=([0,n-1],\Sigma,\delta,0,F)$ be a DFA, $\tau \in \Sigma$, and $i
\in [0,n-1]$. We define $f(A)=|F|$, $f(A,i)=|F\cap [0,i-1]|$. We
denote  by $t_\tau(A, i)$  and $in_\tau(A, i)$ respectively the number of
transitions leaving and reaching $i$. As $t_\tau(A,i)$ is a boolean
function, the complement is $\overline{t}_\tau(A,i) =1-t_\tau(A,i)$. 
Let
$s_\tau(A)=t_\tau(A,0)$,
$a_\tau(A)=\sum_{i\in  F}in_\tau(A,i)$, 
$e_\tau(A)=\sum_{i\in F}t_\tau(A,i)$,
$t_\tau(A)=\sum_{i\in Q}t_\tau(A,i)$, 
$t_\tau(A, [k,l])=\sum_{i\in [k,l]}t_\tau(A,i)$, 
and  the respective complements 
$\overline{s}_\tau(A)=\overline{t}_\tau(A,0)$,
$\overline{e}_\tau(A)= \sum_{i\in F}\overline{t}_\tau(A,i)$, etc.
Whenever there is
no ambiguity we omit $A$ from the above definitions.  All the above
measures, can be defined for a regular language $L$, considering the
measure values for its minimal DFA. 
For instance, we have, $f(L)$, $f(L,i)$, $a_\tau(L)$, $e_\tau(L)$, etc.
We define $s(L)=\sum_{\tau \in
  \Sigma}s_\tau(L)$ and $a(L)=\sum_{\tau \in \Sigma}a_\tau(L)$.

Let $A=([0,n-1],\Sigma,\delta,0,F)$ be a minimal DFA accepting a
finite language, where the states are assumed to be topologically
ordered.  Then, $s(\mathcal{L}(A))= 0$ and there is exactly one final state,
denoted $\pi$ and called \emph{pre-dead}, such that $\sum_{\tau\in
  \Sigma}t_\tau(\pi)= 0.$ The \emph{level} of a state $i$ is the
size of the shortest path from the initial state to $i$, and never
exceeds $n-1$. The level of $A$ is the level of $\pi$.

\section{Union}
\label{sec:union}
Given two incomplete DFAs $A=([0,m-1],\Sigma,\delta_A,0,F_A)$ and
$B=([0,n-1],\Sigma,\delta_B,0,F_B)$ the adaptation of the classical
Cartesian product construction  can be used to obtain a DFA accepting $\mathcal{L}(A)\cup
\mathcal{L}(B)$~\cite{maia13:_incom_trans_compl_of_some_basic_operat}.

\begin{proposition}  
\label{prop:unions}
For any $m$-state incomplete DFA $A$  and any $n$-state incomplete DFA
$B$, both accepting finite languages,  $mn-2$ states are sufficient 
for a DFA accepting $\mathcal{L}(A)\cup
\mathcal{L}(B)$.
\end{proposition}

\begin{proof}
Here we adapt the  proof of  Han and Salomaa~\cite{han08:_state_compl_of_union_and}. 
In the product automaton, the set of states is included in $([0,m-1]\cup \{\Omega_A\})\times([0,n-1]\cup \{\Omega_B\})$,
where $\Omega_A$ and $\Omega_B$ are the dead states of the DFA $A$
and DFA $B$, respectively.  The states of the form $(0,i)$, where
$i\in[1,n-1]\cup\{\Omega_B\}$, and of the form $(j,0)$, where $j\in
[1,m-1]\cup\{\Omega_A\}$, are not reachable from $(0,0)$ because the
operands represent finite languages; the states $(m-1,n-1)$,
$(m-1,\Omega_B)$ and $(\Omega_A,n-1)$ are equivalent because they are
final and they do not have out-transitions; the state
$(\Omega_A,\Omega_B)$ is the dead state and because we are dealing
with incomplete DFAs we can ignore it.  Therefore the number of states
of the union of two incomplete DFAs accepting finite languages is
$(m+1)(n+1) - (m+n) -2 -1= mn-2$.

\end{proof}

\begin{proposition}
\label{prop:uniont}  
For any finite languages $L_1$ and $L_2$ with $isc(L_1)=m$ and
$isc(L_2)=n$, one has
\begin{equation*}
  \begin{aligned}
itc(L_1\cup L_2) \leq & \sum_{\substack{\tau \in
    \Sigma}} \left (s_\tau(L_1) \boxplus s_\tau(L_2) - (itc_\tau(L_1)
  - s_\tau(L_1)) (itc_\tau(L_2)-s_\tau(L_2))\right)\\&
 + n (itc(L_1) -i(L_1))+
    m (itc(L_2)-i(L_2)), 
  \end{aligned}
\end{equation*}
where for $x,y$ boolean values, $x \boxplus y = \min(x+y,1)$.
\end{proposition}

\begin{proof}
In the product automaton, the $\tau$-transitions can be represented as pairs $(\alpha_i,
\beta_j)$ where $\alpha_i$ ($\beta_j$) is $0$ if there exists a
$\tau$-transition leaving the state $i$ ($j$) of  DFA $A$ ($B$), respectively, or
$-1$ otherwise. The resulting DFA can not have transitions of the form
$(-1,-1)$, neither
of the form $(\alpha_0,\beta_j)$, where
$j\in[1,n-1]\cup\{ \Omega_B\}$ nor of the form $(\alpha_i,\beta_0)$,
where $i\in[1,m-1]\cup\{\Omega_A\}$, as happened in the
case of states.  Thus, the  number of
$\tau$-transitions for $\tau \in \Sigma$ are:

\begin{equation*}
\small{
\begin{aligned}
s_\tau(A) \boxplus s_\tau(B)  +& t_\tau(A,[1,m-1])t_\tau(B,[1,n-1])
+ t_\tau(A,[1,m-1]) (\overline{t}_\tau(B,[1,n-1])+1)  
 \\ + &  (\overline{t}_\tau(A, [1,m-1])+1)t_\tau(B, [1,n-1]) \ =\\
s_\tau(A) \boxplus s_\tau(B) +& t_\tau(A, [1,m-1])t_\tau(B, [1,n-1])
+  t_\tau(A, [1,m-1]) (n-t_\tau(B,[1,n-1]))
\\ + &  (m-t_\tau(A, [1,m-1]))t_\tau(B, [1,n-1])\; = \\
s_\tau(A) \boxplus s_\tau(B)  +& n t_\tau(A, [1,m-1]) + mt_\tau(B,
[1,n-1])- t_\tau(A, [1,m-1])t_\tau(B,[1,n-1]).
\end{aligned}
}
  \end{equation*}
  \noindent As the DFAs are minimal, 
  $\sum_{\tau\in\Sigma} t_\tau(A,[1,m-1])$ corresponds to
  $itc(L_1)- s(L_1)$, and analogously for $B$. Therefore the proposition holds.
\end{proof}

Han and Salomaa proved \cite[Lemma
3]{han08:_state_compl_of_union_and} that the upper bound for the
number of states can not be reached with a fixed alphabet.  The
witness families for the incomplete complexities coincide with the
ones already presented for the state
complexity. As we do not consider the dead state, our presentation is
slightly different.  Let $m, n\geq 1$ and $\Sigma =
\{b,c\}\cup \{a_{ij}\mid i\in[1, m-1], j\in [1, n-1],\  (i,j) \ne
(m-1,n-1) \}$.  Let $A=([0,m-1],\Sigma,\delta_A,0,\{m-1\})$ where
$\delta_A(i,b)=i+1$ for $i\in [0, m-2]$ and $\delta_A(0,a_{ij})=i$ for
$j\in [1,n-1],\  (i,j)\ne (m-1,n-1)$.  Let $B=([0,n-1],\Sigma,
\delta_B,0,\{n-1\})$, where $\delta_B(i,c)=i+1$ for $i\in [0, n-1]$
and $\delta_B(0,a_{i,j})=j$ for $j\in [1,n-1], i \in [1, m-1],\  (i,j) \ne (m-1,n-1)$.

\begin{theorem}
For any  integers $m\geq 2$ and $n\geq 2$ there exist an $m$-state 
DFA $A$  and an $n$-state  DFA
$B$, both accepting finite languages, such that any DFA accepting $\mathcal{L}(A)\cup
\mathcal{L}(B)$ needs at least $mn-2$ states and
$3(mn-n-m)+2$ transitions, if the size of the alphabet can depend on
$m$ and $n$.
\end{theorem}

\section{Intersection}
\label{sec:intersection}
Given two incomplete DFAs $A=([0,m-1],\Sigma,\delta_A,0,F_A)$ and
$B=([0,n-1],\Sigma,\delta_B,0,F_B)$ we can obtain a DFA accepting
$\mathcal{L}(A)\cap \mathcal{L}(B)$ by the standard product automaton
construction.

\begin{proposition}
\label{prop:statesint}  
For any $m$-state 
DFA $A$  and any $n$-state 
DFA
$B$, both accepting finite languages,  $mn-2m-2n+6$ states are sufficient 
for a DFA accepting $\mathcal{L}(A)\cap
\mathcal{L}(B)$.
\end{proposition}

\begin{proof}
Consider the DFA  accepting $\mathcal{L}(A)\cap
\mathcal{L}(B)$ obtained by the product construction.
For the same reasons as in Proposition~\ref{prop:unions}, we can eliminate the states of 
the form $(0,j)$, where $j\in[1,n-1]\cup\{\Omega_B\}$,
and of the form $(i,0)$,  where $i\in[1, m-1]\cup\{\Omega_A\}$; 
the states of the form $(m-1,j)$, where $j \in [1,n-2]$, and of
the form  $(i,n-1)$, where $i \in [1,m-2]$ are equivalent to the
state $(m-1,n-1)$ or to the state $(\Omega_A,\Omega_B)$; the states of the form 
$(\Omega_A,j)$, where $j \in [1,n-1]\cup\{\Omega_B\}$, and of the form 
$(i,\Omega_B)$, where $i\in [1,m-1]\cup\{\Omega_A\}$ are equivalent to the state $(\Omega_A,\Omega_B)$
which is the dead state of the DFA resulting fron the intersection,
and  thus can be removed. 
 Therefore, the number of states is 

$(m+1)(n+1)-3((m+1)(n+1))+12-1 = mn -2m -2n +6$.
\end{proof}

\begin{proposition}  
\label{prop:intert}
For any finite languages $L_1$ and $L_2$ with $isc(L_1)=m$ and
$isc(L_2)=n$, one has
\begin{equation*}
  \label{eq:it}
  \begin{split}
itc(L_1\cap L_2) \leq &\sum_{\substack{\tau \in
    \Sigma}}\left( s_\tau(L_1)s_\tau(L_2) +  (itc_\tau(L_1) -
  s_\tau(L_1)\   -\right. \\ & \left. a_\tau(L_1))
    (itc_\tau(L_2) - s_\tau(L_2)- a_\tau(L_2))  +  a_\tau(L_1) a_\tau(L_2)  \right).
  \end{split}
\end{equation*}
\end{proposition}

\begin{proof}
  Using the same technique as in Proposition~\ref{prop:uniont}
  and considering that in the intersection we only have pairs of
  transitions where both elements are different from $-1$, the number of
  $\tau$-transitions is as follows, which proves the proposition,
$$s_\tau(A) s_\tau(B) + t_\tau(A, [1,m-1]\setminus F_A) t_\tau(B,
[1,n-1] \setminus F_B) + a_\tau(A) a_\tau(B).$$
\end{proof}

The witness languages for the tightness of the bounds for this
operation are different from the families
given by Han and Salomaa
because those families are
not tight for the transition complexity.  For $m\geq 2$ and $n\geq 2$,
let $\Sigma = \{a_{ij}\mid i\in [1,m-2],\ j\in [1,n-2]\} \cup
\{a_{m-1,n-1}\}$.  Let $A=([0,m-1],\Sigma, \delta_A,0,\{m-1\})$
where $\delta_A(x, a_{ij})=x+i$ for $x\in[0,m-1],\ i\in [1,m-2]$, and
$j\in [1,n-2]$. Let $B=([0,n-1],\Sigma,
\delta_B,0,\{n-1\})$ where $\delta_B(x, a_{ij})=x+j$ for $x\in
[0,n-1],\  i\in [1,m-2],$ and $j\in [1,n-2]$.

\begin{theorem}
For any  integers $m\geq 2$ and $n\geq 2$ there exist an $m$-state 
DFA $A$  and an $n$-state 
DFA $B$, both accepting finite languages, such that any DFA
accepting $\mathcal{L}(A)\cap \mathcal{L}(B)$ needs at least
$mn-2(m+n)+6$ states and 
$(m-2)(n-2)(2+\sum_{i=1}^{\min(m,n)-3}(m-2-i)(n-2-i))+2$ 
transitions, if
the size of the alphabet can depend on $m$ and $n$.
\end{theorem}
\begin{proof}
  For the number of states, following the proof \cite[Lemma
  6]{han08:_state_compl_of_union_and}, it is easy to see, that the
  words of the set $R= \{\varepsilon\}\cup \{a_{m-1}, a_{n-1}\}\cup
  \{a_{ij}\mid i\in [1, m-2],\text{ and } j\in [1,n-2] \}$ are all
  inequivalent under $\equiv_{L(A)\cap L(B)}$ and $|R|=mn-2(m+n)+6$.

In the DFA $A$, the number of $a_{ij}$-transitions is $(n-2)
\sum_{i=0}^{m-3}(m-1-i)+1$, and in the DFA $B$, that number is $(m-2)
\sum_{i=0}^{n-3}(n-1-i)+1$.  Let $k=(m-2)(n-2)+1$. The DFA resulting
from the intersection operation has: $k$ transitions corresponding to
the pairs of transitions leaving the initial states of the operands;
$(m-2)(n-2)\sum_{i=1}^{\min(m,n)-3}(m-2-i)(n-2-i)$ transitions
corresponding to the pairs of transitions  formed
by transitions leaving non-final and non-initial states of the
operands; and $k$ transitions corresponding to the pairs of
transitions leaving the final states of the operands.
\end{proof}

\section{Complement}
\label{sec:comp}
The state and transition complexity for this operation on finite
languages are similar to the ones on regular languages. This happens
because we need to complete the DFA.
\begin{proposition}  
For any $m$-state 
DFA $A$, accepting a finite language, $m+1$ states are sufficient
for a DFA accepting $\mathcal{L}(A)^c$.
\end{proposition}

\begin{proposition}  
\label{prop:comps}
For any finite languages $L_1$ with $isc(L)=m$ one has $itc(L^c)
\leq|\Sigma|(m+1)$.
\end{proposition}

\begin{proof}
  The maximal number of $\tau$-transitions is $m+1$ because it is the
  number of states. Thus, the maximal number of transitions is
  $|\Sigma| (m+1)$.
\end{proof}

Gao \emph{et al.}~\cite{gao11:_trans_compl_of_incom_dfas} gave the
value $|\Sigma| (itc(L)+2)$ for the transition complexity of the
complement. In some situations, this bound is higher than the bound
here presented, but contrasting to that one, it gives the transition
complexity of the operation as function of the transition complexity
of the operands.

The witness family for this operation is exactly the same 
presented
in the refered paper, i.e.  $\{b^m\}$,
for $m\geq 1$.

\section{Concatenation}
\label{sec:concatenation}

C{\^a}mpeanu \emph{et
  al.}~\cite{campeanu99:_state_compl_of_basic_operat_finit_languag}
studied the state complexity of the concatenation of a $m$-state
complete DFA with a $n$-state complete DFA over an alphabet of size $k$  and proposed the upper
bound

\begin{equation}
  \label{eq:catcamp}
   \sum\limits_{i=0}^{m-2} \min \left \{ k^i,
  \sum\limits_{j=0}^{f(A,i)} \binom{n-2}{j}\right \} +
\min \left \{ k^{m-1},\sum_{j=0}^{f(A)} \binom{n-2}{j}\right \}
\end{equation}

\noindent which was proved to be tight for $m>n-1$. It is easy to see
that  the second term of (\ref{eq:catcamp}) is $\displaystyle{\sum_{j=0}^{f(A)}
  \binom{n-2}{j}}$ if $m>n-1$, and $k^{m-1}$, otherwise. The value
$k^{m-1}$ indicates that the DFA resulting from the concatenation has
states with level at most $m-1$.  But that is not always the case, as
we can see by the example\footnote{
Note that we are omitting the dead state in the figures.
}   in Figure~\ref{fig:Resconcat}. This implies that
(\ref{eq:catcamp}) is not an upper bound if $m<n$. We have

\begin{proposition}
\label{prop:ccats}  
For any $m$-state complete DFA $A$ and any $n$-state complete DFA $B$,
both accepting finite languages over an alphabet of size $k$, the number of states sufficient
for a DFA accepting $\mathcal{L}(A)\mathcal{L}(B)$ is:
\begin{equation}
  \label{eq:catsmgno}
   \sum_{i=0}^{m-2} \min \left \{ k^i,
  \sum_{j=0}^{f(A,i)} \binom{n-2}{j}\right \} +
\sum_{j=0}^{f(A)} \binom{n-2}{j}
\end{equation}
\end{proposition}

In the following, we present tight upper bounds for state and
transition complexity of concatenation for incomplete
DFAs. 

Given two incomplete DFAs $A=([0,m-1],\Sigma, \delta_A, 0, F_A)$ and
$B=([0,n-1], \Sigma, \delta_B,0, F_B)$, that represent finite
languages, the algorithm by Maia \emph{et al.} for the concatenation of regular
languages 
can
be applied to obtain a DFA $C=(R,\Sigma, \delta_C, r_0, F_C)$
accepting $\mathcal{L}(A)\mathcal{L}(B)$.  The set of states of $C$ is
contained in the set $([0,m-1]\cup\{\Omega_A\})\times 2^{[0,m-1]}$,
the initial state $r_0$ is $\langle 0,\emptyset\rangle$ if $0\notin
F_A$, and $\langle 0,\{0\}\rangle$ otherwise; $F_C = \{\langle
i,P\rangle \in R\ \mid \ P \cap F_B \neq \emptyset \}$, and for $\tau
\in \Sigma$, $\delta_C(\langle i, P\rangle, \tau) = \langle i', P'
\rangle$ with $i'=\delta_A(i,\tau)$, if $\delta_A(i,\tau)\downarrow$
or $i'=\Omega_A$ otherwise, and $P'=\delta_B(P,\tau) \cup \{0\}$ if $
i' \in F_A$ and $P'=\delta_B(P,\tau)$ otherwise.

The next result follows the lines of the one presented by C{\^a}mpeanu
\emph{et al.}, with the above referred corrections and omitting the
dead state.

\begin{proposition}
\label{prop:catsmgn}  
For any $m$-state 
DFA $A$  and any $n$-state 
DFA $B$, both accepting finite languages over an alphabet of size $k$, the number of states sufficient 
for a DFA accepting $\mathcal{L}(A)\mathcal{L}(B)$ is:
\begin{equation}
  \label{eq:catsmgn}
   \sum\limits_{i=0}^{m-1} \min \left \{ k^i,
  \sum\limits_{j=0}^{f(A,i)} \binom{n-1}{j}\right \} +
\sum\limits_{j=0}^{f(A)} \binom{n-1}{j} -1.
\end{equation}
\end{proposition}

\begin{proposition}\label{prop:cont1}  
For any finite languages $L_1$ and $L_2$ with $isc(L_1)=m$ and 
$isc(L_2)=n$ over an alphabet of size $k$, and making $\Delta_j=\binom{n-1}{j}
-\binom{\overline{t}_\tau(L_2)-\overline{s}_\tau(L_2)}{j}$, one has

\begin{multline}
itc(L_1L_2) \leq k\sum_{i=0}^{m-2} \min \left \{ k^i,
     \sum_{j=0}^{f(L_1,i)}\binom{n-1}{j} \right \} + \\
+ \sum_{\substack{\tau \in \Sigma}} 
 \left (
  \min  \left \{
  k^{m-1}-\overline{s}_\tau(L_2),\sum_{j=0}^{f(L_1)-1}\Delta_j\right \}+
  \sum_{j=0}^{f(L_1)}\Delta_j
 \right ).
\label{eq:cont1}
\end{multline}
\end{proposition}

\begin{proof}
The $\tau$-transitions of the DFA $C$  accepting
$\mathcal{L}(A) \mathcal{L}(B)$ have three forms: $(i,\beta)$ where $i$ represents the
transition leaving the state $i\in [0,m-1]$; $(-1,\beta)$ where $-1$ represents
the absence of the transition from  state $\pi_A$ to $\Omega_A$;
and $(-2,\beta)$ where $-2$ represents any transition leaving
$\Omega_A$. In all forms, $\beta$ is a set of transitions of DFA $B$.
The number of transitions of the form $(i, \beta)$ is at most
$\sum_{i=0}^{m-2} \min \{ k^i, \sum_{j=0}^{f(L_1,i)} \binom{n-1}{j}\}$
which corresponds to the number of states of the form $(i,P)$, $i \in
[0,m-1]$ and $P\subseteq [0,n-1]$.  The number of transitions of the
form $(-1, \beta)$ is $\min \{
k^{m-1}-\overline{s}_\tau(L_2),\sum_{j=0}^{f(L_1)-1}\Delta_j\}$.
 The size of
$\beta$ is at most $f(L_1)-1$ and we need to exclude the non existing
transitions from non initial states.  On the other hand, we have at
most $k^{m-1}$ states in this level.  However, if 
$s_\tau(B,0)=0$ we need to remove the
transition $(-1,\emptyset)$ which leaves the state $(m-1, \{0\})$.
The number of transitions of the form $(-2,\beta)$ is $
\sum_{j=0}^{f(L_1)}\Delta_j$
and this case is
similar to the previous one.
\end{proof}


To prove that that the bound is reachable we consider  two cases
depending whether $m+1\geq n$ or not.

\subsubsection{Case 1: $m+1\geq n$}
The witness languages are the ones presented by Câmpeanu \emph{et al.} (see
Figure~\ref{fig:Aconcat}). 

\begin{figure}[htb]
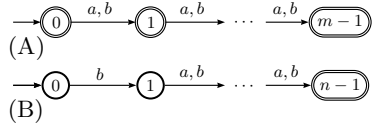

\centering
  \begin{tabular}[h]{lc}
{\small(A)} & {\centering
 \TinyPicture\VCDraw{
\begin{VCPicture}{(10,-2)(0,0)}
\FinalState[0]{(0,0)}{A} \FinalState[1]{(3,0)}{B} 
\FinalStateVar[m-1]{(9,0)}{D} 
\SetStateLineStyle{none}
\FinalStateVar[\cdots]{(6,0)}{C} 
\Initial{A} 
\EdgeL{A}{B}{a,b}
\EdgeL{B}{C}{a,b}
\EdgeL{C}{D}{a,b}
\end{VCPicture}
}}
\\
{\small(B)} & {\centering
\TinyPicture\VCDraw{
\begin{VCPicture}{(10,-2)(0,0)}
\State[0]{(0,0)}{A} \State[1]{(3,0)}{B} 
\FinalStateVar[n-1]{(9,0)}{D}
\SetStateLineStyle{none}
\FinalStateVar[\cdots]{(6,0)}{C} 
\Initial{A} 
\EdgeL{A}{B}{b}
\EdgeL{B}{C}{a,b}
\EdgeL{C}{D}{a,b}
\end{VCPicture}
}}
\end{tabular}
\caption{DFA $A$ with $m$ states and DFA $B$ with $n$ states.}
\label{fig:Aconcat}
\end{figure}

\begin{figure}[htb]
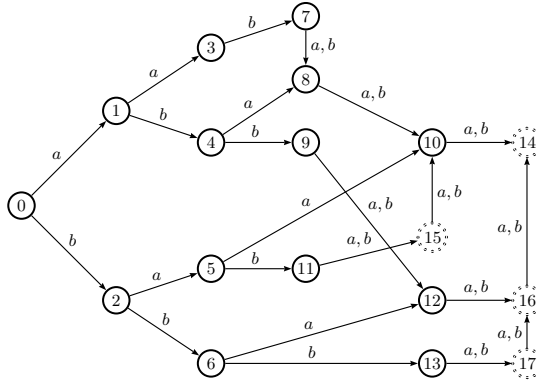

\centering
\TinyPicture\VCDraw{
\begin{VCPicture}{(16,-6)(0,6)}
\State[0]{(0,0)}{0} 
\State[1]{(3,3)}{1} 
\State[2]{(3,-3)}{2}
\State[3]{(6,5)}{3}  
\State[4]{(6,2)}{4}  
\State[5]{(6,-2)}{5}  
\State[6]{(6,-5)}{6}  
\State[7]{(9,6)}{7}  
\State[8]{(9,4)}{8}  
\State[9]{(9,2)}{9}  
\State[10]{(13,2)}{10}  
\State[11]{(9,-2)}{11}  
\State[12]{(13,-3)}{12}  
\State[13]{(13,-5)}{13}
\SetStateLineStyle{dotted}
\FinalState[14]{(16,2)}{14}    
\FinalState[15]{(13,-1)}{15}    
\FinalState[16]{(16,-3)}{16}    
\FinalState[17]{(16,-5)}{17}    
\Initial{A} 
\EdgeL{0}{1}{a}
\EdgeL{0}{2}{b}
\EdgeL{1}{3}{a}
\EdgeL{1}{4}{b}
\EdgeL{2}{5}{a}
\EdgeL{2}{6}{b}
\EdgeL{3}{7}{b}
\EdgeL{4}{8}{a}
\EdgeL{4}{9}{b}
\EdgeL{5}{10}{a}
\EdgeL{5}{11}{b}
\EdgeL{6}{12}{a}
\EdgeL{6}{13}{b}
\EdgeL{7}{8}{a,b}
\EdgeL{8}{10}{a,b}
\EdgeL{9}{12}{a,b}
\EdgeL{10}{14}{a,b}
\EdgeL{11}{15}{a,b}
\EdgeL{12}{16}{a,b}
\EdgeL{13}{17}{a,b}
\EdgeR{15}{10}{a,b}
\EdgeL{16}{14}{a,b}
\EdgeL{17}{16}{a,b}
\end{VCPicture}}
\caption{DFA resulting of the concatenation of DFA $A$ with $m=3$ 
and DFA $B$ with $n=5$, of Fig.~\ref{fig:Aconcat}. The states
with dashed lines have level $>3$ and are not accounted for by formula
(\ref{eq:catcamp}).}
\label{fig:Resconcat}
\end{figure}

\begin{theorem}
For any  integers $m\geq 2$ and $n\geq 2$ there exist an $m$-state 
DFA $A$  and an $n$-state DFA
$B$, both accepting finite languages, such that any DFA accepting
$\mathcal{L}(A)\mathcal{L}(B)$ needs at least 
$(m-n+3)2^{n-1}-2$ states
and $6\cdot 2^{n-1} -8$ transitions, if $m+1 \geq n$. 
\end{theorem}
\begin{proof}
  The proof for the number of states corresponds to the one presented 
by C{\^a}mpeanu
\emph{et al.}.
The DFA
  $A$ has $m-1$ $\tau$-transitions for $\tau\in\{a,b\}$ and $f(A)=m$.
  The DFA $B$ has $n-2$ $a$-transitions and $n-1$ $b$-transitions.
  Consider $m \geq n$. If we analyse the transitions as we did in the
  proof of the Proposition~\ref{prop:cont1} we have: $2^{n-1}-1$
  $a$-transitions and $2^{n-1}-1$ $b$-transitions that correspond to
  the transitions of the form $(i, \beta)$; $2^{n-1}-2$
  $a$-transitions and $2^{n-1}-1$ $b$-transitions that correspond to
  the transitions of the form $(-1,\beta)$; and $2^{n-1}-2$
  $a$-transitions and $2^{n-1}-1$ $b$-transitions that correspond to
  the transitions of the form $(-2,\beta)$. Adding up those values we
  have the result.
\end{proof}

\subsubsection{Case 2: $m+1 < n$}
Let $\Sigma= \{b\} \cup \{a_i \mid i\in [1,n-2]\}$. 
Let $A=([0,m-1], \Sigma, \delta_A, 0, [0,m-1])$ where
$\delta_A(i, \tau)=i+1$, for any $\tau \in \Sigma$.
Let $B=([0,n-1], \Sigma, \delta_B, 0, \{n-1\})$ where
$\delta_B(i, b)=i+1$, for $i\in [0,n-2]$, 
$\delta_B(i, a_j)=i+j$, for $i,j\in [1,n-2]$, $i+j\in [2,n-1]$, and
$\delta_B(0, a_j)=j$, for $j\in [2,n-2]$.

\begin{theorem}
For any  integers $m\geq 2$ and $n\geq 2$ there exist an $m$-state 
DFA $A$  and an $n$-state DFA
$B$, both accepting finite languages, such that the number of
states and transitions of any DFA accepting
$\mathcal{L}(A)\mathcal{L}(B)$ reaches the upper bounds, if $m+1 <n $
and the size of the alphabet can depend of $m$ and $n$. 
\end{theorem}

\begin{proof}
  The number of $\tau$-transitions of DFA $A$ is $m-1$, for $\tau \in
  \Sigma$. The DFA $B$ has $n-1$ $b$-transitions, $n-2$
  $a_1$-transitions, and $n-i$ $a_i$-transitions, with $i \in
  [2,n-2]$. The proof is similar
  to the proof of Proposition~\ref{prop:cont1}.
\end{proof}

\begin{proposition}
The upper bounds for state and transition complexity of
concatenation cannot be reached with a fixed alphabet for $m\geq 0$,
$n> m+1$.
\end{proposition}
\begin{proof}
Let $S=\{(\Omega_A,P)\mid 1\in P\}\subseteq R$. 
  A state   $(\Omega_A,P)\in S$  has to
  satisfy the following condition: $\exists i\in F_A \exists
  P'\subseteq 2^{[0,n-1]}$ with  $0\in P'$ and $\exists \tau\in \Sigma$,
    such that $\delta_C((i,P'),\tau)=(\Omega_A,P)$. The  maximal
    size of $S$ is  $\sum_{j=0}^{f(A)-1} \binom{n-2}{j}$. Assume that $\Sigma$ has a fixed size $k=|\Sigma|$. Then, the maximal
    number    of words that reaches states of $S$ from $r_0$ is
    $\sum_{i=0}^{f(A)}k^{i+1}$. 
It is easy to see that for $n>m$
    sufficiently large $\sum_{i=0}^{f(A)}k^{i+1}\ll \sum_{j=0}^{f(A)-1} \binom{n-2}{j}$.
\end{proof}

\section{Star}
\label{sec:star}

Given an incomplete DFA $A=([0,m-1],\Sigma, \delta_A, 0, F_A)$
accepting a finite language,  
a DFA $B$ accepting $\mathcal{L}(A)^\star$ can be constructed by an
algorithm similar to the one for regular languages \cite{maia13:_incom_trans_compl_of_some_basic_operat}.
Let $B=(Q_B,\Sigma,\delta_B,\{0\},F_B)$ where $Q_B\subseteq
2^{[0,m-1]}$, $F_B=\{P\in Q_B\mid P\cap F_A\not=\emptyset\}\cup\{0\}$,
and for $\tau\in \Sigma$,  $P\subseteq Q_B$, and
$R=\delta_A(P,\tau)$, $\delta_B(P,\tau)$ is $R$ if $R\cap
F_A=\emptyset$, $R\cup \{0\}$ otherwise. 

If $f(A)=1$ then the minimal DFA accepting $\mathcal{L}(A)^\star$ has
also $m$ states. Thus, in the following we will consider DFAs with at
least two final states.

\begin{proposition}  
\label{prop:stars}
For any $m$-state
DFA $A$ accepting a finite language with $f(A)\geq 2$, $2^{m-f(A)-1}+2^{m-2}-1$
states are sufficient for a DFA accepting $\mathcal{L}(A)^\star$.
\end{proposition}

\begin{proof}
The proof is similar to the proof presented by C{\^a}mpeanu
\emph{et al.}.
\end{proof}

\begin{proposition}  
\label{prop:transtar}
For any finite language $L$ with $isc(L)=m$ one has
$$ itc(L^\star) \leq 2^{m-f(L)-1}\left (k +\sum_{\tau
  \in\Sigma}2^{e_\tau(L)}\right)- \sum_{\tau \in \Sigma} 2^{n_\tau}-  \sum_{\tau \in X} 2^{n_\tau}$$
\noindent where $n_\tau=\overline{t}_\tau(L)- \overline{s}_\tau(L)-\overline{e}_\tau(L)$ and $X=\{\tau\in \Sigma\mid s_\tau(L)=0\}$.

\end{proposition}

\begin{proof}
The proof is similar to the one for the states. 
\end{proof}
The witness  family for this operation is the same as the one presented
by Câmpeanu \emph{et al.}, but we have to exclude dead state
 (see Figure~\ref{fig:Astar}).  

\begin{figure}[htb]
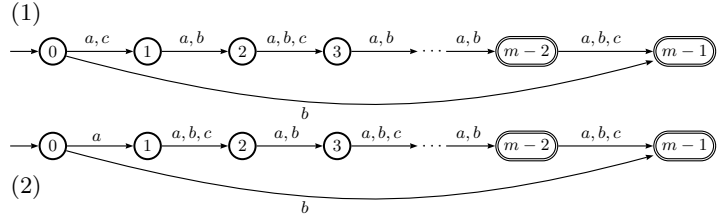

\centering
  \begin{tabular}[h]{lc}
{\small(1)} & {\centering
 \TinyPicture\VCDraw{
\begin{VCPicture}{(0,0)(20,2)}
\State[0]{(0,0)}{A} \State[1]{(3,0)}{B} \State[2]{(6,0)}{C}
\State[3]{(9,0)}{E} 
\FinalStateVar[m-2]{(15,0)}{F} 
\FinalStateVar[m-1]{(20,0)}{G} 
\SetStateLineStyle{none}
\State[\cdots]{(12,0)}{D} 
\Initial{A} 
\EdgeL{A}{B}{a,c}
\VArcR{arcangle=-15}{A}{G}{b}
\EdgeL{B}{C}{a,b}
\EdgeL{C}{E}{a,b,c}
\EdgeL{E}{D}{a,b}
\EdgeL{D}{F}{a,b}
\EdgeL{F}{G}{a,b,c}
\end{VCPicture}
}}
\\
{\small(2)} & {\centering
 \TinyPicture\VCDraw{
\begin{VCPicture}{(20,3)(0,0)}
\State[0]{(0,0)}{A} \State[1]{(3,0)}{B} \State[2]{(6,0)}{C}
\State[3]{(9,0)}{E} 
\FinalStateVar[m-2]{(15,0)}{F} 
\FinalStateVar[m-1]{(20,0)}{G} 
\SetStateLineStyle{none}
\State[\cdots]{(12,0)}{D} 
\Initial{A} 
\EdgeL{A}{B}{a}
\VArcR{arcangle=-15}{A}{G}{b}
\EdgeL{B}{C}{a,b,c}
\EdgeL{C}{E}{a,b}
\EdgeL{E}{D}{a,b,c}
\EdgeL{D}{F}{a,b}
\EdgeL{F}{G}{a,b,c}
\end{VCPicture}
}}
\end{tabular}
\caption{DFA $A$ with $m$ states, with $m$ even (1) and odd (2).}
\label{fig:Astar}
\end{figure}

\begin{theorem}
\label{teo:starw}
For any  integer $m\geq 4$ there exists an $m$-state 
DFA $A$
accepting a finite language, such that any DFA accepting
$\mathcal{L}(A)^\star$ 
needs at least $2^{m-2}+2^{m-3}-1$ states and 
$9\cdot 2^{m-3}-2^{m/2}-2$ transitions if $m$ is odd or
$9\cdot 2^{m-3}-2^{(m-2)/2}-2$ transitions otherwise. 
\end{theorem}

\section{Reversal}
\label{sec:reverse}

Given an incomplete DFA $A=([0,m-1],\Sigma, \delta_A, 0, F_A)$, to
obtain a DFA $B$ that accepts $\mathcal{L}(A)^R$,  we  first
reverse all transitions of $A$ and then determinize the resulting NFA.


\begin{proposition}  
  For any $m$-state 
  DFA $A$, with $m \leq 3$, accepting a finite language over an
  alphabet of size $k\geq 2$, $\sum_{i=0}^{l-1} k^i + 2^{m-l}-1$ states
  are sufficient for a DFA accepting $\mathcal{L}(A)^R$, where $l$ is
  the smallest integer such that $2^{m-l}\leq k^l$.
\end{proposition}

\begin{proof}
  The proof is similar to the proof of~\cite[Theorem
  5]{campeanu99:_state_compl_of_basic_operat_finit_languag}. We only
  need to remove the dead state.
\end{proof}

\begin{proposition}  
\label{prop:revtran}
For any finite language $L$ with $isc(L)=m$  and if  $l$ is
  the smallest integer such that $2^{m-l}\leq k^l$, one has, if $m$ is odd,
$$itc(L^{R}) \leq \sum_{i=0}^{l} k^i-1+k  2^{m-l} -\sum_{\tau \in \Sigma}2^{\sum_{i=0}^{l-1}\overline{t}_\tau(L,i)+1},$$
\noindent or, if $m$ is even, 
$$itc(L^{R}) \leq \sum_{i=0}^{l} k^i-1+k  2^{m-l} -\sum_{\tau
  \in \Sigma}\left (2^{\sum_{i=0}^{l-2}\overline{t}_\tau(L,i)+1}-c_\tau(l)\right),$$
\noindent where $c_\tau(l)$ is $0$ if there exists a $\tau$-transition reaching the state $l$ and $1$ otherwise.

\end{proposition}

\begin{proof}
  The smallest $l$ that satisfies $2^{m-l}\leq k^l$ is the same for
  $m$ and $m+1$, and because of that we have to consider whether $m$
  is even or odd.  Suppose $m$ odd. Let $T_1$ be set of transitions
  corresponding to the first $\sum_{i=0}^{l-1} k^i$ states and $T_2$
  the set corresponding to the other $2^{m-l}-1$ states.  We have that
  $|T_1|=\sum_{i=0}^{l-1} k^i-1$, because the initial state has no
  transition reaching it.  As the states of DFA $B$ for the reversal
  are sets of states of DFA $A$ we also consider each
  $\tau$-transition as a set. If all $\tau$-transitions were defined
  its number in $T_2$ would be $2^{m-l}$.  Note that the transitions
  of the $m-l$ states correspond to the transitions of the states
  between $0$ and $l-1$ in the initial DFA $A$, thus we remove the
  sets that has no  $\tau$-transitions. As the initial
  state of $A$ has no transitions reaching it, we need to add one to the number
  of missing $\tau$-transitions.  Thus, $|T_2|=\sum_{\tau \in
    \Sigma}2^{m-l}-2^{(\sum_{i=0}^{l-1}(\overline{t}_\tau(i)))+1}$.

Let us consider $m$ even. In this case we need also to consider the
set of transitions that connect  the states with the  highest level in the first set with the
states with the lowest level in the second set.  As the 
 highest level is $l-1$, we have to remove the possible transitions
 that reach the state $l$ in DFA $A$.
\end{proof}
The witness family for this operation is the one presented by Câmpeanu
\emph{et al.} but we omit the dead state (see Figure~\ref{fig:Arev}).

\begin{figure}[htb]
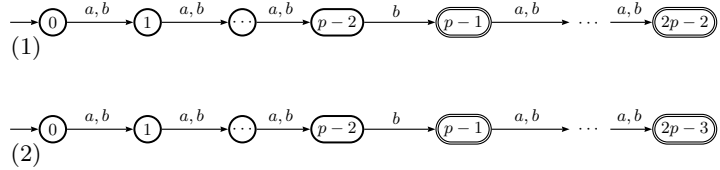

\centering
  \begin{tabular}[h]{lc}
{\small(1)} & {\centering
 \TinyPicture\VCDraw{
\begin{VCPicture}{(20,2)(0,0)}
\State[0]{(0,0)}{A} \State[1]{(3,0)}{B} \State[\cdots]{(6,0)}{C}
\StateVar[p-2]{(9,0)}{D} 
\FinalStateVar[p-1]{(13,0)}{E} 
\FinalStateVar[2p-2]{(20,0)}{G} 
\SetStateLineStyle{none}
\FinalStateVar[\cdots]{(17,0)}{F} 
\Initial{A} 
\EdgeL{A}{B}{a,b}
\EdgeL{B}{C}{a,b}
\EdgeL{C}{D}{a,b}
\EdgeL{D}{E}{b}
\EdgeL{E}{F}{a,b}
\EdgeL{F}{G}{a,b}
\end{VCPicture}
}}
\\
{\small(2)} & {\centering
 \TinyPicture\VCDraw{
\begin{VCPicture}{(20,2)(0,0)}
\State[0]{(0,0)}{A} \State[1]{(3,0)}{B} \State[\cdots]{(6,0)}{C}
\StateVar[p-2]{(9,0)}{D} 
\FinalStateVar[p-1]{(13,0)}{E} 
\FinalStateVar[2p-3]{(20,0)}{G} 
\SetStateLineStyle{none}
\FinalStateVar[\cdots]{(17,0)}{F} 
\Initial{A} 
\EdgeL{A}{B}{a,b}
\EdgeL{B}{C}{a,b}
\EdgeL{C}{D}{a,b}
\EdgeL{D}{E}{b}
\EdgeL{E}{F}{a,b}
\EdgeL{F}{G}{a,b}
\end{VCPicture}
}}
\end{tabular}
\caption{DFA $A$ with $m=2p-1$ states (1) and  with $m=2p$ (2).}
\label{fig:Arev}
\end{figure}

\begin{theorem}
  For any integer $m\geq 4$ there exists an $m$-state DFA $A$
  accepting a finite language, such that any DFA accepting
  $\mathcal{L}(A)^R$ needs at least $3\cdot 2^{p-1}+2$ states and
  $3\cdot 2^{p}-8$ transitions if $m=2p-1$ or $2^{p+1}-2$ states and
  $2^{p+2}-7$ transitions if $m=2p$.
\end{theorem}

\section{Final Remarks}
\label{sec:conclusao}
In this paper we studied the incomplete state and transition
complexity of basic regularity preserving operations on finite
languages.
Table~\ref{table:wtc} summarizes some of those results.
For unary finite languages the incomplete transition complexity  is
equal to the incomplete state complexity of that language, which is always
equal to the state complexity of the language minus one.  

As future work we plan to study the average transition complexity of
these operations following the lines of Bassino \emph{et al.}~\cite{bassino10:_averag_state_compl_of_ration}.

\vspace{-0.3cm}

\bibliography{fa1203}{}

\begin{thebibliography}{10}
\providecommand{\url}[1]{\texttt{#1}}
\providecommand{\urlprefix}{URL }

\bibitem{bassino10:_averag_state_compl_of_ration}
Bassino, F., Giambruno, L., Nicaud, C.: The average state complexity of
  rational operations on finite languages. Int. J. Found. Comput. Sci.  21(4),
  495--516 (2010)

\bibitem{beesley03:_finit_state_morph}
Beesley, K.R., Karttunen, L.: Finite State Morphology. CSLI Publications,
  Stanford University (2003)

\bibitem{campeanu99:_state_compl_of_basic_operat_finit_languag}
C{\^a}mpeanu, C., II, K.C., Salomaa, K., Yu, S.: State complexity of basic
  operations on finite languages. In: Boldt, O., J{\"u}rgensen, H. (eds.) WIA
  1999. LNCS, vol. 2214, pp. 60--70. Springer (2001)

\bibitem{cassandras06:_introd_to_discr_event_system}
Cassandras, C.G., Lafortune, S.: Introduction to discrete event systems.
  Springer (2006)

\bibitem{gao11:_trans_compl_of_incom_dfas}
Gao, Y., Salomaa, K., Yu, S.: Transition complexity of incomplete {DFA}s.
  Fundam. Inform.  110(1-4),  143--158 (2011)

\bibitem{han08:_state_compl_of_union_and}
Han, Y.S., Salomaa, K.: State complexity of union and intersection of finite
  languages. Int. J. Found. Comput. Sci.  19(3),  581--595 (2008)

\bibitem{hopcroft79:_introd_to_autom_theor_languag_and_comput}
Hopcroft, J.E., Ullman, J.D.: Introduction to Automata Theory, Languages and
  Computation. Addison-Wesley (1979)

\bibitem{maia13:_incom_trans_compl_of_some_basic_operat}
Maia, E., Moreira, N., Reis, R.: Incomplete transition complexity of some basic
  operations. In: van Emde~et al., P. (ed.) SOFSEM 2013. LNCS, vol. 7741, pp.
  319--331. Springer (2013)

\bibitem{maurel05:_autom_and_diction}
Maurel, D., Guenthner, F.: Automata and Dictionaries. College Publications
  (2005)

\bibitem{owens09:_regul_expres_deriv_re_examin}
Owens, S., Reppy, J.H., Turon, A.: Regular-expression derivatives re-examined.
  J. Funct. Program.  19(2),  173--190 (2009)

\bibitem{salomaa97:_nfa_to_dfa_trans_for}
Salomaa, K., Yu, S.: {NFA} to {DFA} transformation for finite languages over
  arbitrary alphabets. J. of Aut., Lang. and Comb.  2(3),  177--186 (1997)

\bibitem{shallit08:_secon_cours_in_formal_languag}
Shallit, J.: A Second Course in Formal Languages and Automata Theory. CUP
  (2008)

\bibitem{yu97:_handb_formal_languag}
Yu, S.: Regular languages. In: Rozenberg, G., Salomaa, A. (eds.) Handbook of
  {F}ormal {L}anguages, vol.~1, pp. 41--110. Springer (1997)

\end{thebibliography}

\bibliographystyle{splncs03}

\end{document}